\documentclass[11pt]{article}
\pdfoutput=1 

\usepackage[margin=1in]{geometry}

\usepackage{mathtools}
\usepackage{amsfonts, amsmath, amssymb, amsthm}
\usepackage{xcolor}
\usepackage{bbm}

\usepackage{complexity}
\usepackage[hyphens]{url}
\usepackage{lmodern}
\usepackage{nicefrac}
\usepackage{thmtools}
\usepackage{microtype}
\usepackage[colorlinks=true]{hyperref}
\hypersetup{
	linkcolor=[rgb]{0.2,0.2,0.5},
	citecolor=[rgb]{0.2, 0.5, 0.2},
	urlcolor=[rgb]{0.5, 0.2, 0.2}
}

\renewcommand{\L}{\ComplexityFont{L}}

\newcommand{\LATER}[1]{\PackageWarning{miforbes}{#1}}

\DeclareMathOperator{\tr}{tr}
\newcommand{\eqdef}{:=}
\newcommand{\defeq}{=:}

	\numberwithin{equation}{section}
	\declaretheoremstyle[bodyfont=\it,qed=$\qedsymbol$]{noproofstyle} 
	\declaretheoremstyle[bodyfont=\it,qed=$\lozenge$]{defstyle} 
	\declaretheoremstyle[qed=$\lozenge$]{rmkstyle} 

	\theoremstyle{plain}

	\declaretheorem[name=Theorem,numberlike=equation]{theorem}
	\declaretheorem[name=Theorem,numberlike=equation,style=noproofstyle]{theoremwp}
	\declaretheorem[name=Theorem,unnumbered]{theorem*}
	\declaretheorem[name=Theorem,unnumbered,style=noproofstyle]{theoremwp*}

	\declaretheorem[name=Lemma,numberlike=equation]{lemma}
	
	\declaretheorem[name=Lemma,unnumbered]{lemma*}
	\declaretheorem[name=Lemma,unnumbered,style=noproofstyle]{lemmawp*}

	\declaretheorem[name=Corollary,numberlike=equation]{corollary}
	
	\declaretheorem[name=Corollary,unnumbered]{corollary*}
	\declaretheorem[name=Corollary,unnumbered,style=noproofstyle]{corollarywp*}

	\declaretheorem[name=Proposition,numberlike=equation]{proposition}
	
	\declaretheorem[name=Proposition,unnumbered]{proposition*}
	\declaretheorem[name=Proposition,unnumbered,style=noproofstyle]{propositionwp*}

	\declaretheorem[name=Claim,unnumbered]{claim*}
	\declaretheorem[name=Claim,unnumbered,style=noproofstyle]{claimwp*}

	\declaretheorem[name=Subclaim,unnumbered]{subclaim*}
	\declaretheorem[name=Subclaim,unnumbered,style=noproofstyle]{subclaimwp*}

	\declaretheorem[name=Theorem,numberlike=equation,style=noproofstyle]{theorem-cited}

	\declaretheorem[name=Definition,style=defstyle,numberlike=equation]{definition}
	\declaretheorem[name=Definition,style=defstyle,unnumbered]{definition*}

	\declaretheorem[name=Conjecture,style=defstyle,unnumbered]{conjecture*}

	\declaretheorem[name=Construction,style=defstyle,unnumbered]{construction*}

	\declaretheorem[name=Open Question,style=defstyle,unnumbered]{open*}

	\declaretheorem[name=Remark,style=rmkstyle,unnumbered]{remark*}

	\declaretheorem[name=Example,style=rmkstyle,unnumbered]{example*}

	\declaretheorem[name=Notation,style=rmkstyle,unnumbered]{notation*}

	\declaretheorem[name=Question,style=rmkstyle,unnumbered]{question*}

\newcommand{\F}{\mathbb{F}}  
  
\renewcommand{\R}{\mathbb{R}}  

\newcommand{\N}{\mathbb{N}}  
     
\renewcommand{\E}{\mathbb{E}}  
\renewcommand{\P}{\Pr} 

\DeclareSymbolFont{bbold}{U}{bbold}{m}{n}
\DeclareSymbolFontAlphabet{\mathbbold}{bbold}
\newcommand{\1}{\mathbbold{1}}

\newcommand{\ip}[2]{\left\langle#1 ,#2 \right\rangle}
\newcommand{\abs}[1]{\left|#1 \right|}
\newcommand{\norm}[1]{\|#1 \|}
\newcommand{\paren}[1]{\left(#1 \right)}

\newcommand{\ind}[1]{\1 \! \! \paren{#1}}
\newcommand{\prob}[1]{\P \! \paren{#1}}

\newcommand{\bits}{\{0,1\}}
\newcommand{\ch}{\chi_{\alpha}}
\newcommand{\chb}{\chi_{\beta}}

\renewcommand{\hat}{\widehat}
\newcommand{\Otilde}{\widetilde{O}}

\begin{document}

\title{Pseudorandom Generators for Read-Once Branching Programs, in any Order}
\date{August 19, 2018}
\author{Michael A. Forbes~\thanks{Email: \texttt{miforbes@illinois.edu}. Department of Computer Science, University of Illinois at Urbana-Champaign. Supported by NSF grant CCF-1755921.}
	\and
	Zander Kelley~\thanks{Email: \texttt{awk2@illinois.edu}. Department of Computer Science, University of Illinois at Urbana-Champaign. Supported by NSF grant CCF-1755921.}
}
\maketitle

\begin{abstract}
	A central question in derandomization is whether randomized logspace ($\RL$) equals deterministic logspace ($\L$). To show that $\RL=\L$, it suffices to construct explicit pseudorandom generators (PRGs) that fool polynomial-size read-once (oblivious) branching programs (roBPs).  Starting with the work of Nisan~\cite{n92}, pseudorandom generators with seed-length $O(\log^2 n)$ were constructed (see also \cite{inw94, gr14}).  Unfortunately, improving on this seed-length in general has proven challenging and seems to require new ideas.  

	A recent line of inquiry (e.g., \cite{bv10,gmrtv12,imz12,rsv13,svw14,hlv17,lv17,chrt17}) has suggested focusing on a particular limitation of the existing PRGs (\cite{n92,inw94,gr14}), which is that they only fool roBPs when the variables are read in a particular \emph{known} order, such as $x_1<\cdots<x_n$. In comparison, existentially one can obtain logarithmic seed-length for fooling the set of polynomial-size roBPs that read the variables under any fixed \emph{unknown} permutation $x_{\pi(1)}<\cdots<x_{\pi(n)}$.  While recent works have established novel PRGs in this setting for subclasses of roBPs, there were no known $n^{o(1)}$ seed-length explicit PRGs for general polynomial-size roBPs in this setting.

	In this work, we follow the ``bounded independence plus noise'' paradigm of Haramaty, Lee and Viola~\cite{hlv17,lv17}, and give an improved analysis in the general roBP unknown-order setting.  With this analysis we obtain an explicit PRG with seed-length $O(\log^3 n)$ for polynomial-size roBPs reading their bits in an unknown order.  Plugging in a recent Fourier tail bound of Chattopadhyay, Hatami, Reingold, and Tal~\cite{chrt17}, we can obtain a $\Otilde(\log^2 n)$ seed-length when the roBP is of constant width.
\end{abstract}

\section{Introduction}

\LATER{}

A central goal in complexity theory is to understand the power of randomness in computation, in particular the $\mathsf{P}$ vs $\BPP$ problem. A particularly natural method of showing $\mathsf{P}=\BPP$ is to construct an explicit $\varepsilon$-error pseudorandom generator (PRG) with sufficiently small seed-length $\ell$, ideally logarithmic. That is, a function $G:\bits^\ell\to\bits^n$ such that for any sufficiently efficiently computable $f$,
\[
	\left|\underset{y \in \bits^\ell}\E f(G(y)) - \underset{x \in \bits^n}\E f(x)\right|
	\le\varepsilon
	\;.
\]
Given such a PRG, one can then replace the randomness of a $\BPP$ algorithm with the pseudorandom output and then enumerate over all such seeds to obtain a deterministic algorithm by majority vote (if $\varepsilon$ is a sufficiently small constant). After decades of work, the hardness-vs-randomness paradigm (see for example Vadhan~\cite{v12}) shows that the construction of pseudorandom generators fooling general polynomial-size circuits is intimately tied to the quest for circuit lower bounds, which remain out of reach.  As such, a long line of work has sought to derandomize subclasses of $\BPP$. A particularly fruitful model to study has been randomized logspace ($\RL$), as not only do PRGs for $\RL$ have natural applications, but they can also be unconditionally constructed, for example as done in the seminal work of Nisan~\cite{n92}.

In particular, Nisan~\cite{n92} constructed a PRG fooling the non-uniform version of $\RL$, that is, the class of polynomial-size read-once (oblivious) branching programs (roBPs). A read-once branching program can be thought of as a finite automaton that takes in binary input strings $x$ of some fixed size $n$.
Additionally, the transition function of the automaton is allowed to depend on the position $i$ of each bit $x_i$.
We say that the branching program has width $w$ if the each layer of time the finite automaton has $w$ states.
Visually, branching programs can be represented as a layered acyclic digraph with $n+1$ layers, each containing $w$ nodes; the transition function is then represented by assigning two outgoing edges at each interior node into the next layer. The existence of a logspace-computable PRG $G : \bits^{\ell(n)} \rightarrow \bits^n$ for branching programs of width $w = n^{O(1)}$ is sufficient to show that $\BPL \subseteq \DSPACE(\ell(n))$.

\LATER{}
\LATER{}

Nisan~\cite{n92} gave a construction of a PRG with seed-length $\ell = O(\log^2n)$ for polynomial-width roBPs. 
Since then, there have been various constructions (\cite{inw94,gr14})  recovering the same seed-length using different techniques, but there has been little quantitative progress towards the desired seed-length $\ell = O(\log n)$.\footnote{By this we mean quantitative progress in the constant-error regime. Recently in \cite{bcg17}, Braverman, Cohen, and Garg give a hitting set (a ``one-sided'' PRG) with a better seed-length in the small-error regime than Nisan's generator.}
In fact, it remains open even to achieve a seed-length of $\ell = O(\log^2 n / \log \log n)$, even for constant-width branching programs.

The constructions of Nisan and (\cite{inw94,gr14}) all employ a common high-level approach which can be summarized by the following ``communication'' argument. The first half of a branching program can communicate with the second half only via the state reached in the middle layer.
Since there are only $w$ states in this layer, the second half of the program should ``learn'' roughly only $\log w$ bits of information about the input bits fed to the first half.
Because of this, it is safe to reuse all but roughly $\log w$ of the bits of entropy invested to generate the first half of the input string to generate the second half. 
This argument is then applied recursively to the left and right subprograms.

\LATER{}
\LATER{}

There is some feeling that this particular recursive paradigm will not yield generators with seed-lengths better than $O(\log^2 n)$ (\cite{bv10,rsv13,svw14}), and that new, more flexible techniques are required to make progress.
A crucial feature of this paradigm is that the PRG knows the order in which its pseudorandom output will be read.
In fact, it is known that Nisan's generator fails to generate pseudorandom strings that fool branching programs if they read the bits of the string in a different order than anticipated (\cite{t09}). 
The search for a different paradigm motivates the following challenge: construct a PRG that fools branching programs which may read their input in any order.
To formalize this, we define the notion of an unknown-order roBP: a function $g : \bits^n \rightarrow \bits$ of the form $g(x) = f( x_{\pi(1)}, x_{\pi(2)}, \ldots, x_{\pi(n)} )$, where $\pi$ is a permutation (independent of $x$) and $f$ is a roBP.

\LATER{}

Bogdanov, Papakonstantinou, and Wan~\cite{bpw11} constructed a PRG with seed-length $(1 - \Omega(1)) \cdot n$ for unknown-order roBP of width $w = n^{O(1)}$.
Their primary motivation for doing so was to derive the first generator with nontrivial seed-length that fools read-once formulas. 
Read-once formulas can be simulated by small-width read-once branching programs for \emph{some} order $\pi$, and hence existing generators for \emph{known}-order roBPs (\cite{n92,inw94,gr14}) would not suffice.
Impagliazzo, Meka, and Zuckerman~\cite{imz12} achieved a generator with seed length $\ell = (nw)^{1/2 + o(1)}$ for unknown-order roBP of width $w$.\footnote{
In fact, their generator fools the more general model of branching programs that may read the input bits any number of times and in any adaptive order.}

\section{Our Work}

Here, we give the first PRG with poly-logarithmic seed-length for $\poly(n)$-width unknown-order roBPs.

\begin{theorem}\label{main1}
	There exists an explicit $\nicefrac{1}{\poly(n)}$-error pseudorandom generator $G : \bits^{O( \log^3 n)} \rightarrow \bits^n$ for the class of functions computable by a $\poly(n)$-width read-once (oblivious) branching program in some variable order.
\end{theorem}

As a corollary, we also derive the first PRG with poly-logarithmic seed length for read-once formulas (see \cite{bpw11} for the reduction).

\begin{corollarywp*} 
	There exists an explicit pseudorandom generator $G : \bits^{O( \log^3 n)} \rightarrow \bits^n$ for read-once formulas with constant fan-in. 
\end{corollarywp*}

\subsection{Our Techniques}

We now briefly describe our proof technique at a high-level, with a more technical discussion given in \autoref{strategy}.  The main motivation comes from the ``bounded independence plus noise'' paradigm introduced by Haramaty, Lee, and Viola (\cite{hlv17,lv17}).   There, they study the addition (modulo 2) of a low-wise independent distribution with a pseudorandom noise distribution.  The intuition is that to fool a function $f$, it suffices to create a distribution to dampen all non-constant Fourier coefficients.  For low-degree Fourier coefficients, this can be achieved by a low-wise independent distribution.  In the other extreme, high-degree Fourier coefficients are dampened by coordinate-wise independent noise.  The addition of these two distributions can then inherit the best of both distributions and fool the desired function $f$.

However, the above outline has two challenges.  First, the noise distribution (picking each coordinate independently amongst $\{0,1\}$) requires too large a seed-length.  To address this, the work of Haramaty, Lee, and Viola (\cite{hlv17,lv17}) proposed to use a pseudorandom noise distribution where a pseudorandom set of coordinates are first chosen, and then the elements within those coordinates are then substituted with \emph{truly} random values.  While this proposal as stated still requires a large seed-length, the key observation is that the number of truly random bits has shrunk from $n$ originally to $\approx \nicefrac{n}{2}$ (for if you choose a (pseudorandom) subset of $\{1,\ldots,n\}$ it has size $\approx \nicefrac{n}{2}$).  Thus, one can hope to then recursively apply the construction in $\approx \log n$ rounds until no random bits are further required.

The second, more serious, challenge is to show that a single step of ``bounded independence plus (pseudorandom) noise'' actually fools the target function $f$, and doing so is the main contribution of this work (\autoref{prooflem2}).  The difficulty in addressing this is that there are too many high-degree Fourier coefficients, so that while each can be individually fooled by the construction, we cannot apply a union bound while maintaining a small seed-length.  Indeed, nothing so far in this discussion has used anything about the structure of the function $f$, which clearly must be used to obtain a small seed-length.

To meet this challenge, we avoid a naive union bound by instead grouping high-degree Fourier coefficients into a small number of groups, each of which can be dampened at once.  Specifically, we group high-degree Fourier coefficients into $n$ sets, where the $i$-th set contains those coefficients that ``become'' high-degree upon reading the $i$-th variable (\autoref{prop1}).\footnote{Note that this grouping depends on the order of the variables. However, this grouping only occurs in the \emph{analysis}, and that the construction itself is oblivious to the variable order.} On an intuitive level, one can then appeal to the ``bounded communication'' aspect of a roBP to argue that in the $i$-th grouping those variables read after the $i$-th variable can be essentially ignored. We are then left with Fourier coefficients that are of \emph{medium} degree (for if they were very-high degree they would have been put in the $j$-th group for some $j<i$).  The number of such medium-degree Fourier coefficients is not too large (because the degree is not too large), and yet each such coefficient is dampened by the noise distribution (because the degree is not too small). This then allows us to apply a union bound to obtain that we have dampened all the Fourier coefficients in the $i$-th grouping, and by applying this for all $i$ we obtain the result.

\subsection{The Constant-Width Case}

Although we give the first PRG with poly-logarithmic seed-length for the general case of poly-width unknown-order roBPs, such a seed-length has been qualitatively achieved in the constant-width case as a result of a recent line of work.
Reingold, Steinke, and Vadhan~\cite{rsv13} gave a PRG with seed-length $O(\log^2 n)$ for unknown-order \emph{permutation} branching programs of constant width.
Steinke, Vadhan, and Wan~\cite{svw14} gave a PRG with seed-length $\widetilde{O}(\log^3 n)$ for unknown-order width-3 branching programs. 
Chattopadhyay, Hatami, Reingold, and Tal~\cite{chrt17} gave a PRG with seed-length $\widetilde{O}(\log^{w+1} n)$ for unknown-order branching programs of constant-width $w$.
Central to each of these results is a bound on a certain key quantity: the level-$k$ Fourier mass of a branching program (see \autoref{prelim}).
In each work, a bound on this quantity is established for the class of branching programs under consideration, and then this bound is used to deduce the result.

Although we also employ a Fourier analytic approach, a major contrast between our techniques and this line of work is that in general we have no need of any nontrivial bound on the Fourier mass of branching programs.
However, we can still make use of one to replace otherwise naive bounds on Fourier mass in our argument.
By incorporating the level-$k$ Fourier mass bound for constant-width branching programs derived in \cite{chrt17} into our approach, we get the following improvement on \autoref{main1} in the constant-width case.

\begin{theorem}\label{main2}
	There exists an explicit $\nicefrac{1}{\poly(n)}$-error pseudorandom generator $G : \bits^{\Otilde( \log^2 n)} \rightarrow \bits^n$ for the class of functions computable by a $O(1)$-width read-once (oblivious) branching program in some variable order.
\end{theorem}

Thus in the constant-width case, we nearly recover the $O( \log^2 n)$ seed-length of Nisan's generator for the more challenging model of unknown-order branching programs.

\section{Preliminaries} \label{prelim}

Here we describe a convenient algebraic encoding of a branching program as a product of one-bit matrix-valued functions.
Recall that a branching program of width $w$ is a $w$-state finite automaton where the transition map is allowed to depend on the number of bits read so far.
Let us encode the $w$ states as the set of standard basis vectors in $\R^{w}$. 
Then, the transition map corresponding to the $i$-th input bit $x_i$ can be encoded by a pair of transition matrices $A_{i,0}, A_{i,1} \in \R^{w \times w}$, defined so that $A_{i,x_i}$ applied to the current state produces the appropriate successor state.
Define the one-bit matrix-valued functions $F_i(x_i) = A_{i,x_i}$. 
With this notation in place, the value of a branching program $f$ on an input $x \in \bits^n$ is given by the $(1,1)$ entry of the product
\[
	F(x) \eqdef F_1(x_1) F_2(x_2) \cdots F_n(x_n)
	\;.
\]
This entry indicates whether the string $x$ defines a path through the program that takes the start state to the accepting state.

Let $U$ denote the uniform distribution over $\bits^n$ and let $X$ be an arbitrary distribution over $\bits^n$.
To show that a branching program $f$ is $\varepsilon$-fooled by $X$, it suffices to bound the error $\abs{ \E_X f(X) - \E_U f(U) } = \abs{ \E_X F(X)_{1,1} - \E_U F(U)_{1,1} }$ by $\varepsilon$.
However, it will be more convenient to simply bound the Frobenius norm of the entire error matrix 
\[
	\E_X F(X) - \E_U F(U) 
	\;,
\]
by $\varepsilon$.
Recall that the Frobenius norm of a matrix is defined by 
\[
	\norm{M} \eqdef \sqrt{\tr\paren{ M^\top M}} = \sqrt{ \sum_{i,j} M_{i,j}^2 }
	\;,
\]
so clearly such a bound is also sufficient.
In this paper, $\norm{\cdot}$ will always denote the Frobenius norm.

\subsection{Fourier Analysis}

For every vector $\alpha \in \F_2^n$, define the associated Fourier character $\ch : \F_2^n \rightarrow \R$ via
$$ \ch(x) = (-1)^{\ip{\alpha}{x}} .$$
We say that $\ch$ is a degree-$k$ Fourier coefficient if $\abs{\alpha} = k$, where $\abs{\alpha}$ denotes the hamming weight of $\alpha$.
Any function $F : \F_2^n \rightarrow \R^{w \times w}$ can be expanded in the basis of Fourier characters, with coefficients from the ring $\R^{w \times w}$.
The Fourier expansion of $F$ is 
$$ F(x) = \sum_{\alpha \in \F_2^n} \hat{F}_{\alpha} \ch(x),$$
where the Fourier coefficients $\hat{F}_{\alpha}$ are given by
$$ \hat{F}_\alpha \eqdef \underset{x \in \F_2^n}\E F(x) \ch(x). $$
This identity can easily be checked with the aid of a few useful properties of Fourier characters; namely
$$ \ch(x) \chb(x)  = \chi_{\alpha + \beta}(x) $$
and 
\[
	\underset{x \in \F_2^n}\E \ch(x) = 
		\begin{cases}
			1	&\text{ if } \alpha = 0 \\
			0	&\text{ otherwise}.
		\end{cases}
\]

Aside from this, we will only need a few simple facts about Fourier analysis. 
Firstly, we note that the expectation of a function $F$ under the uniform distribution is conveniently encoded by its $0$-th degree Fourier coefficient:
$$ \E_U F(U) = \sum_{\alpha \in \F_2^n} \hat{F}_{\alpha} \E_U \ch(U) = \hat{F}_0 .$$
Next, we will need Parseval's identity to get a bound on the sum of squares of Fourier coefficients.

\begin{proposition}[Parseval for matrix-valued functions]
	\[
		\sum_{\alpha \in \F_2^n} \norm{\hat{F}_\alpha}^2 = \underset{x \in \F_2^n}\E \norm{ F(x) }^2
		\;.
	\]
\end{proposition}
\begin{proof}
	Let $\ip{\cdot}{\cdot}$ denote the Frobenius matrix inner product, which is defined by
	$$ \ip{M}{N} = \tr \! \paren{M^\top N} = \sum_{i,j} M_{i,j} N_{i,j} . $$
	First note that
	\begin{align*}
		\norm{F(x)}^2
		&= \ip{F(x)}{ F(x) } \\
		&= \ip{ \sum_{\alpha \in \F_2^n} \hat{F}_{\alpha} \ch(x)}{ \sum_{\beta \in \F_2^n} \hat{F}_{\beta} \chb(x) } \\
		&= \sum_{\alpha \in \F_2^n} \sum_{\beta \in \F_2^n} \ip{ \hat{F}_{\alpha}}{ \hat{F}_{\beta}} \chi_{\alpha + \beta}(x)
		\;.
	\end{align*}
	Thus
	\begin{align*}
		\E_{x \in \F_2^n} \norm{ F(x) }^2 
		&=\sum_{\alpha \in \F_2^n} \sum_{\beta \in \F_2^n} \ip{ \hat{F}_{\alpha}}{ \hat{F}_{\beta}} \underset{x \in \F_2^n}\E \chi_{\alpha + \beta}(x) \\
		&=\sum_{\alpha \in \F_2^n} \ip{ \hat{F}_{\alpha}}{ \hat{F}_{\alpha}} \\
		&=\sum_{\alpha \in \F_2^n} \norm{ \hat{F}_{\alpha}}^2
		\;.
		\qedhere
	\end{align*}
\end{proof}

We remark that if $F$ is a branching program, then upon any input $x$, $F(x)$ is equal to some transition matrix that has exactly $w$ entries of value $1$ and its remaining entries are all zeros. 
Thus for branching programs we have $\norm{F(x)}^2 = w$ for any $x$, and the above identity gives $\sum_{\alpha \in \F_2^n} \norm{\hat{F}_{\alpha}}^2 = w$.

Finally, we define $\mathcal{L}_k(F)$, the level-$k$ Fourier mass of a function, as in \cite{rsv13}:
\[
	\mathcal{L}_k(F) \eqdef  \sum_{\substack{\alpha \in \F_2^n \\ |\alpha| = k}} \norm{ \hat{F}_{\alpha} }
	\;.
\]
Recalling that $\norm{ \hat{F}_{\alpha} } = \norm{ \E_x F(x) \ch(x)} \leq \E_x \norm{  F(x) } = w^{1/2}$, note that we have the trivial bound 
\[
	\mathcal{L}_k(F) \leq \binom{n}{k} w^{1/2}
	\;,
\]
for branching programs.\footnote{Actually, one can easily derive the slightly better bound $\mathcal{L}_k(F) \leq \binom{n}{k}^{1/2} w^{1/2}$ by first applying Cauchy-Schwarz followed by Parseval's identity.}

Chattopadhyay, Hatami, Reingold, and Tal~\cite{chrt17} derive the following bound on the level-$k$ Fourier mass of a branching program which significantly improves upon the trivial bound in the case of small $w$.

\begin{theoremwp}[Chattopadhyay, Hatami, Reingold, and Tal~\cite{chrt17}]
	Suppose $F : \F^2_n \rightarrow \R^{w \times w}$ encodes a branching program of width $w$. Then
	\[
		\mathcal{L}_k(F) \leq O( \log n)^{wk}
		\;.
		\qedhere
	\]
\end{theoremwp}

We also define the level-$k$ Fourier complexity of width-$w$ branching programs,
\[
	\mathcal{L}(n,w;k) \eqdef \max_{F} \sum_{i=1}^k \mathcal{L}_i(F)
	\;,
\]
where the maximum is taken over all functions $F : \F_2^n \to \R^{w \times w}$ that encode a branching program. 

\subsection{Pseudorandom Primitives}

In this section, we collect some basic pseudorandom distributions over $\F_2^n$ that will serve as building blocks for our generator.
It is important that the defining property of each of these distributions remains unaffected by a re-ordering of bit positions.
The fact that we build our generator out of permutation-invariant components is what allows it to fool branching programs that read their input in any order.

First we introduce $\delta$-biased distributions, which fool linear functions over $\F_2^n$, i.e.\ Fourier characters.
\begin{definition} Let $D$ be a distribution over $\F_2^n$. We say $D$ is $\delta$-biased if, for every nonzero $\alpha \in \N_2^n$, we have
	\[
		\abs{ \E_D \ch(D) } \leq \delta
		\;.
		\qedhere
	\]
\end{definition}
It is possible to sample from a $\delta$-biased distribution using $O(\log n + \log\nicefrac{1}{\delta})$ random bits (\cite{nn93,aghp92}).

Next we have $k$-wise independent and $\gamma$-almost $k$-wise independent distributions, which look locally uniform and thus fool functions that only depend on a few bits.

\begin{definition} Let $D$ be a distribution over $\F_2^n$. We say $D$ is $k$-wise independent if, for every $f : \F_2 \rightarrow [-1, 1]$ that depends on at most $k$ bits, we have
$$ \E_D f(D) = \E_U f(U). $$
If $D$ merely satisfies 
$$ \abs{ \E_D f(D) - \E_U f(U) } \leq \gamma $$
for every such $f$, we say that $D$ is $\gamma$-almost $k$-wise independent.
\end{definition}
It is possible to sample from a $k$-wise independent distribution using $O(k \cdot \log n)$ random bits (\cite{v12}) and from a $\gamma$-almost $k$-wise independent distribution using $O(k + \log \log n + \log 1/\gamma)$ random bits (\cite{nn93,aghp92}).
We remark that Fourier characters $\ch(x)$ only depend on $|\alpha|$ bits of $x$, so these distributions also fool low-degree Fourier characters. 

\section{The Generator}

We adopt our construction from the ``bounded independence plus noise'' framework developed by Haramaty, Lee, and Viola in \cite{hlv17,lv17}.
In fact this framework is essentially equivalent to the ``mild pseudorandom restriction'' framework developed by Gopalan, Meka, Reingold, Trevisan and Vadhan~\cite{gmrtv12} and subsequently employed by various authors (e.g., \cite{rsv13,svw14,chrt17}), but we find the bounded independence plus noise perspective more convenient to work with.

We actually give two slightly different constructions. The first construction only uses $k$-wise independence as a core pseudorandom primitive (along with appropriate recursion), and suffices for proving \autoref{main1}. The second construction replaces the use of \emph{exact} $k$-wise independence with the use small-bias spaces and \emph{almost} $k$-wise independence.  This allows the error analysis to be more general (but also slightly more involved due to additional parameters).  In particular, this second construction also suffices for handing general roBPs (\autoref{main1}), but now the added generality can be tuned to achieve a better seed-length for constant-width roBPS, as required to prove \autoref{main2}.

We proceed to give the first construction.

Let $D_1, D_2, \ldots, D_r$ denote $r$ independent copies of a $2k$-wise independent distribution and let $T_1, T_2, \ldots, T_r$ denote $r$ independent copies of a $k$-wise independent distribution over $\F_2^n$.
We then define the distribution $G_r$ recursively as follows. 
Let $G_0$ be equal to the all-ones string in $\F_2^n$, and set
$$ G_{i+1} \eqdef D_i + T_i \land G_i,$$
where $\land$ denotes bitwise AND and $+$ denotes addition over $\F_2^n$ (i.e.\ bitwise XOR).

\begin{lemma} \label{mainlem1}
Suppose $F : \F_2^n \rightarrow \R^{w \times w}$ encodes a branching program. Then $G_r$, with parameters $k = \lceil 5  \lg n + 2 \lg w \rceil$ and $r =  \lceil 2 \lg n + \frac{1}{2} \lg w \rceil$, fools $F$ with error
$$ \varepsilon = \norm{ \E_{G_r} F(G_r) - \E_U F(U) } \leq O\paren{\tfrac{1}{n}} .$$
\end{lemma}

This proves \autoref{main1}, since in the case of width $w = n^{O(1)}$ the price of sampling such a distribution is $O(r \cdot k \cdot \log n) = O(\log^3 n)$ random bits.

We now define the second construction, $G_r^*$. 
This time, let $D_1, D_2, \ldots, D_r$ denote $r$ independent $\delta$-biased distributions, and let $T_1, T_2, \ldots, T_r$ denote $r$ independent $\gamma$-almost $k$-wise independent distributions.
We set $G_0^*$ equal to a $320k$-wise independent distribution, and again we define
$$ G_{i+1}^* \eqdef D_i + T_i \land G_i^*. $$

\begin{lemma} \label{mainlem2}
Suppose $F : \F_2^n \rightarrow \R^{w \times w}$ encodes a branching program.
Then $G_{r}^*$, with depth $r = \lceil \lg n \rceil$, fools $F$ with error
$$ \varepsilon = \norm{ \E_{G_r^*} F(G_r^*) - \E_U F(U) } \leq O\paren{ \left( \sqrt{\delta} \mathcal{L}(n,w;k) + \left(\tfrac{1}{2}\right)^{k/2} + \sqrt{\gamma} + \gamma 4^k \right) \cdot nwr} .$$
\end{lemma}

Since $G_r^*$ can be sampled at the cost of 
$$ O \big( \left(\log n + \log 1/\delta + k + \log1/\gamma \right) \cdot r + k\cdot \log n \big)$$
random bits, it suffices to set 
\begin{itemize}
\item $r =  \lceil \lg n \rceil$
\item $k = \lceil 3 \lg(nw/\varepsilon) \rceil$
\item $\gamma = (nw / \varepsilon)^{-9}$
\item $\delta = (nw  \mathcal{L}(n,w;k)  / \varepsilon)^{-3}$
\end{itemize}
to get a generator with seed length
$$ \ell = O \big( \left(\log(nw/\epsilon) +   \mathcal{L}(n,w;k)  \right) \cdot \log n \big)$$
that $O(\varepsilon)$-fools branching programs $F$ of width $w$.
From this we derive the following two corollaries by invoking either the trivial bound or the bound from \cite{chrt17} on the level-$k$ Fourier mass of width-$w$ branching programs.
The second of these corollaries proves \autoref{main2}.

\begin{corollary}
	There exists an explicit PRG with seed length
	$$ \ell = O \big( \log(nw/\varepsilon) \cdot \log^2 n \big) $$
	that $\varepsilon$-fools unknown-order branching programs of width $w$.
\end{corollary}

\begin{corollary}
	There exists an explicit PRG with seed length
	$$ \ell = O \big( w \cdot \log(nw/\varepsilon) \cdot \log n \cdot \log \log n \big) $$
	that $\varepsilon$-fools unknown-order branching programs of width $w$.
\end{corollary}

\section{Proof Strategy}\label{strategy}

We show that $G_r$ successfully fools branching programs with the following inductive analysis.
By adding and subtracting the term $F(D_i + T_i \land U)$, we have 

\begin{align*}
\norm{\E_{G_i} F(G_i) - \E_U F(U)}
&= \norm{\E_{D_i} \E_{T_i} \E_{G_{i-1}} F(D_i + T_i \land G_{i-1}) - \E_U F(U)} \\\
&\leq \norm{\E_{D_i} \E_{T_i} \E_{U} F(D_i + T_i \land U) - \E_U F(U)}  \\
& \; +  \E_{D_i} \E_{T_i} \norm{\E_{G_{i-1}} F(D_i + T_i \land G_{i-1}) - \E_{U} F(D_i + T_i \land U)}.
\end{align*}

Since, for any fixed vectors $d,t \in \F_2^n$, the function $F'(x) : = F(d + t \land x)$ is again some branching program, we argue that $F'$ is fooled inductively.
The bulk of our proof is then spent arguing that $\E_{D,T,U} F(D + T \land U) \approx \E_U F(U)$.

The starting point of this argument is the observation that the ``noise-like'' distribution $T \land U$ successfully fools any function that is divisible by a high-degree Fourier character.
Specifically, suppose that $T$ is $k$-wise independent and that $\alpha \in \F_2^n$ has hamming weight $|\alpha| \geq k$.
Let $g : \F_2^n \rightarrow [-1, 1]$ be an arbitrary function such that $g$ and $\ch$ depend on disjoint sets of input bits.
Then $T \land U$ fools the function $f \eqdef \ch \cdot g$ with error $\varepsilon = \nicefrac{1}{2^k}$:
\[
	\E_U f(U) = \paren{ \E_U \ch(U) } \paren{ \E_U g(U) } = 0
	\;,
\]
and
\begin{align*}
	|\E_{T,U} f(T \land U)|
	&=\left| \E_T \big( \E_U \ch(T \land U) \big) \big(\E_U g(U)  \big) \right|\\
	&\leq \E_T| \E_U \ch(T \land U) | | \E_U g(U) |\\
	&\leq \E_T |\E_U  \ch(T \land U)  | \\
	&=\E_T \ind{ \alpha \land T = 0} \\
	&\leq \frac{1}{2^k}
	\;.
\end{align*}

Overall, we wish to enact the following plan. 
If $F$ encodes our branching program, we wish to use Fourier analysis to rewrite $F$ as a sum of simpler terms, and then use linearity of expectation together with a triangle inequality to argue that $D+T \land U$ fools each term separately.\footnote{
This approach is inspired by the similar arguments of Haramaty, Lee, and Viola employed in \cite{hlv17}. However, the generators they produce with this idea have seed length $\geq n^{1/2}$, while we achieve generators with poly-logarithmic seed-length. }
Recall the product structure of $F$,
$$ F(x) = F_1(x_1) F_2(x_2) \cdots F_n(x_n), $$
where $F_i(x_i) = A_{i, x_i}$.
We can imagine taking the Fourier expansions of these one-bit factors:
\[
	F_i(x_i) = \tfrac{1}{2} (A_{i,0} + A_{i,1}) + \tfrac{1}{2} (A_{i,0} - A_{i,1})  (-1)^{x_i} \defeq B_{i,0} + B_{i,1} (-1)^{x_i}
	\;,
\]
so that $F$ has the form
\[
	F(x) = \prod_{i=1}^n \big( B_{i,0} + B_{i,1} (-1)^{x_i} \big)
	\;.
\]
If we expand this product completely, we recover the Fourier expansion of $F$.
Certainly the terms in this sum are simple enough, but the problem is that there are too many of them; we cannot afford a $2^n$-fold triangle inequality.
Instead, we expand this product more ``slowly'': only until we see that some term has collected a degree-$k$ Fourier character as a factor.
By the above observation, this term can be killed immediately, and we go on expanding the remaining terms.
At the end, we are left with only low-degree Fourier coefficients, which can all be fooled by the distribution $D$.
By doing this carefully, we get by with only $n$ applications of the triangle inequality.

\section{Proof of \autoref{mainlem1}}

Suppose 
\[
	F(x) = F_1 (x_1) F_2(x_2) \cdots F_n(x_n)
\]
encodes a branching program. We define the subprograms of $F$,
\[
	F^{\leq i}(x_1,\ldots,x_i)\eqdef F_1(x_1) F_2(x_2) \cdots F_i(x_i)
\]
and 
\[
	F^{> i}(x_{i+1},\ldots,x_n) \eqdef F_{i+1}(x_{i+1}) F_{i+2}(x_{i+2}) \cdots F_{n}(x_n)
	\;,
\]
so that $F = F^{\leq i} \cdot F^{> i}$ for any $i$ (we define $F^{> n}$ as the empty product, so that it is an identity matrix $I$). 
With this notation in place, we can re-express $F$ in the following convenient form.

\begin{proposition} \label{prop1}
	Suppose $F : \F_2^n \rightarrow \R^{w \times w}$ encodes a branching program. Then for any $k\ge 1$, $F$ can be written as
	$$ F = \hat{F}_{0} + L + \sum_{i = 1}^n H_i \cdot F^{> i}, $$
	where
	$$ L = \sum_{\substack{\alpha \in \F_2^n \\ 0 < |\alpha| < k}} \hat{F}_\alpha \chi_\alpha,$$
	and
	$$ H_i = \sum_{\substack{\alpha \in \F_2^i \\ |\alpha| = k \\ \alpha_i = 1}} \hat{F}^{\leq i}_\alpha \chi_\alpha.$$
\end{proposition}

\begin{proof}
	We verify that the expression has the same Fourier expansion as $F$. 
	\begin{align*}
		\sum_{i=1}^n H_i \cdot F^{>i}
		&=	\sum_{i=1}^n 
			\Bigg(  \sum_{\substack{\alpha \in \F_2^i \\ |\alpha| = k \\ \alpha_i = 1}} \hat{F}^{\leq i}_\alpha \chi_\alpha \Bigg)
			\Bigg(  \sum_{\beta \in \F_2^{n-i}} \hat{F}^{> i}_\beta \chi_\beta\Bigg) \\
		&= 	\sum_{i=1}^n 
			\sum_{\substack{\alpha \in \F_2^i \\ |\alpha| = k \\ \alpha_i = 1}}
			\sum_{\beta \in \F_2^{n-i}}
			\hat{F}_{\alpha \beta} \chi_{\alpha \beta} \\
		&= 	\sum_{i=1}^n 
			\sum_{\substack{\alpha \in \F_2^n \\ \alpha_1 + \alpha_2 + \cdots + \alpha_i = k \\ \alpha_i = 1}}
			\hat{F}_{\alpha} \chi_{\alpha}(x) \\
		&= 	\sum_{\substack{\alpha \in \F_2^n \\ |\alpha| \geq  k}}
			\hat{F}_{\alpha} \chi_{\alpha}
			\;.
			\qedhere
	\end{align*}
\end{proof}

Now we derive a useful expectation bound for functions whose Fourier expansions are only supported at degree $k$.

\begin{lemma} \label{prooflem1}
	Let
	\[
		H(x) = \sum_{\substack{\alpha \in \F_2^n \\ |\alpha| = k}} \hat{H}_{\alpha} \ch(x)
	\]
	be some function whose Fourier expansion is supported only at degree $k$. 
	Let $D$, $T$, and $U$ denote respectively a $2k$-wise independent, a $k$-wise independent, and a uniform distribution over $\F_2^n$.
	Then we have
	\[
		\E_{D,T} \| \E_U H(D + T \land U) \|^2 
		\leq \frac{1}{2^k}\sum_{|\alpha| = k} \|\hat{H}_\alpha \|^2 
		\;.
	\]
\end{lemma}
\begin{proof}
	Firstly, note that
	\begin{align*}
		\E_U H(D + T \land U) 
		&=\E_U \sum_{|\alpha| = k} \hat{H}_\alpha \cdot \chi_\alpha(D + T \land U) \\
		&=\E_U \sum_{|\alpha| = k} \hat{H}_\alpha \cdot \chi_\alpha(D)  \cdot \chi_\alpha(T \land U) \\
		&=\sum_{|\alpha| = k} \hat{H}_\alpha \cdot \chi_\alpha(D) \cdot \E_U \chi_\alpha(T \land U) \\
		&=\sum_{|\alpha| = k} \hat{H}_\alpha \cdot \chi_\alpha(D) \cdot \ind{\alpha \land T = 0}
		\;.
	\end{align*}

	Letting $\ip{\cdot}{\cdot}$ denote the Frobenius matrix inner product, we have
	\begin{align*}
		\| \E_U H(D + T \land U) \|^2 
		&=  \ip{ \E_U H(D + T\land U)}{ \E_U H(D + T\land U) }\\
		&=  \ip{  \sum_{|\alpha| = k} \hat{H}_\alpha \chi_\alpha(D) \ind{\alpha \land T = 0}}{ \sum_{|\beta| = k} \hat{H}_\beta \chi_\beta(D) \ind{\beta \land T = 0} }\\
		&= \sum_{|\alpha| = k} \sum_{|\beta| = k} \ip{ \hat{H}_\alpha}{ \hat{H}_\beta} \cdot \chi_{\alpha + \beta}(D) \cdot \ind{(\alpha \lor \beta) \land T  = 0}
	\end{align*}
	and
	\begin{align*}
		\E_{D} \| \E_U H(D + T \land U) \|^2 
		&= \E_D \sum_{|\alpha| = k} \sum_{|\beta| = k} \ip{ \hat{H}_\alpha}{ \hat{H}_\beta} \chi_{\alpha + \beta}(D) \ind{(\alpha \lor \beta) \land T  = 0} \\
		&= \sum_{|\alpha| = k} \sum_{|\beta| = k} \ip{ \hat{H}_\alpha}{ \hat{H}_\beta} \Big( \E_D \chi_{\alpha + \beta}(D) \Big) \ind{ (\alpha \lor \beta) \land T = 0} \\
		&= \sum_{|\alpha| = k} \ip{ \hat{H}_\alpha}{ \hat{H}_\alpha} \ind{\alpha \land T = 0} \\
		&= \sum_{|\alpha| = k} \norm{\hat{H}_\alpha}^2 \ind{\alpha \land T = 0}
		\;.
	\end{align*}
	Finally,
	\begin{align*}
		\E_{D,T} \| \E_U H(D + T \land U) \|^2 
		&= \E_T \E_D \| \E_U H(D + T \land U) \|^2  \\
		&= \E_T \sum_{|\alpha| = k} \norm{\hat{H}_\alpha}^2 \ind{\alpha \land T = 0}\\
		&= \sum_{|\alpha| = k} \norm{\hat{H}_\alpha}^2 \cdot \E_T \ind{\alpha \land T = 0} \\
		&=\sum_{|\alpha| = k} \norm{\hat{H}_\alpha}^2 \cdot \left( \tfrac{1}{2} \right)^k
		\;.
		\qedhere
	\end{align*}
\end{proof}

We now have the tools in place to prove our main technical lemma.

\begin{lemma} \label{prooflem2}
	Suppose $F : \F_2^n \rightarrow \R^{w \times w}$ encodes a branching program.
	Let $D$, $T$, and $U$  denote respectively a $2k$-wise independent, a $k$-wise independent, and a uniform distribution over $\F_2^n$.
	Then $D + T \land U$ fools $F$ with error
	$$ \varepsilon =  \norm{\E_{D,T,U}F(D + T \land U) - \E_U F(U)} \leq \frac{n w}{2^{k/2}} .$$
\end{lemma}

\begin{proof}
	To analyze $\| \E_{D,T,U} F(D + T \land U) - \E_{U} F(U) \|$, we use the preceding expansion of $F$ together with linearity of expectation and the triangle inequality.
	Recalling that $\E_U F(U) = \hat{F}_0$, this gives
	\begin{align*}
		\| \E_{D,T,U} F(D + T \land U) - \E_{U} F(U) \| &\leq \\
		\| \E_{D,T,U} L(D + T \land U)  \|& + \sum_{i=1}^n \| \E_{D,T,U} H_i(D + T \land U) F^{> i}(D + T \land U) \|
		\;.
	\end{align*}
	The low-degree term is dealt with easily by $D$:
	\begin{align*}
		\E_{D,T,U} L(D + T \land U)
		&=\E_D \E_T \E_U \sum_{\substack{\alpha \in \F_2^n \\ 0 < |\alpha| < k}} \hat{F}_\alpha \chi_\alpha(D + T \land U) \\
		&=\E_D \E_T \E_U \sum_{\substack{\alpha \in \F_2^n \\ 0 < |\alpha| < k}} \hat{F}_\alpha \ch(D) \ch(T \land U) \\
		&=\E_T \E_U \sum_{\substack{\alpha \in \F_2^n \\ 0 < |\alpha| < k}} \hat{F}_\alpha \cdot \Big( \E_D \ch(D) \Big) \cdot \ch(T \land U) \\
		&=0
		\;.
	\end{align*}
	Now for each $i$ we have
	\begin{align*}
		\| \E_{D,T,U} H_i(D + T \land U) F^{> i}(D + T \land U) \| 
		&=\left\| \E_{D,T} \big(\E_U H_i(D + T \land U) \big) \big(\E_U F^{> i} (D + T \land U) \big) \right\|\\
		&\le\E_{D,T} \| \E_U H_i(D + T \land U) \|  \|\E_U F^{> i} (D + T \land U) \| \\
		&\le\E_{D,T} \| \E_U H_i(D + T \land U) \| w^{1/2} \\
		&\le\Big( \E_{D,T} \|\E_U H_i(D + T \land U) \|^2 \Big)^{1/2} w^{1/2} \\
		&\le\left( \tfrac{1}{2} \right)^{k/2} \Big( \sum_{\alpha \in \F_2^n} \norm{\hat{F}^{\leq i}_\alpha}^2 \Big)^{1/2} w^{1/2}\\
		&=\left( \tfrac{1}{2} \right)^{k/2} w,
	\end{align*}
	where we get the final equality by applying the Parseval identity to $F^{\leq i}$.
\end{proof}

\subsection{Proof of \autoref{mainlem1}}

\begin{proof}
	Recall the induction framework outlined in \autoref{strategy}:

	\begin{align*}
		\norm{\E_{G_i} F(G_i) - \E_U F(U)}
		&= \norm{\E_{D_i} \E_{T_i} \E_{G_{i-1}} F(D_i + T_i \land G_{i-1}) - \E_U F(U)} \\
		&\leq \norm{\E_{D_i} \E_{T_i} \E_{U} F(D_i + T_i \land U) - \E_U F(U)}  \\
		& \; +  \E_{D_i} \E_{T_i} \norm{\E_{G_{i-1}} F(D_i + T_i \land G_{i-1}) - \E_{U} F(D_i + T_i \land U)}.
	\end{align*}

	We have seen how to carry out the inductive step; it remains to establish the base case.
	To do this, we wish to think of $F(G_r)$ as a function of $G_0$ only, with $D_1, D_2, \ldots, D_r$ and $T_1, T_2, \ldots, T_r$ fixed. 

	Specifically, we do the following. 
	Define the strings $g_0 \eqdef x$ and $g_{i} \eqdef D_i + T_i \land g_{i-1}$, and define the function $f(x) \eqdef F(g_r)$.
	Note that with this setup we have $F(G_r) = f(G_0)$, and so

	\begin{align*}
	\norm{\E_{G_r} F(G_r) - \E_U F(U)}
	&= \norm{\E_{D_r} \E_{T_r} \E_{G_{r-1}} F(D_r + T_r \land G_{r-1}) - \E_U F(U)} \\
	&\leq \frac{nw}{2^{k/2}} +  \E_{D_r} \E_{T_r} \norm{\E_{G_{r-1}} F(D_r + T_r \land G_{r-1}) - \E_{U} F(D_r + T_r \land U)} \\
	&\leq  \frac{nw}{2^{k/2}} + \frac{nw}{2^{k/2}} + \cdots + \frac{nw}{2^{k/2}}+ \underset{\substack{D_1, D_2, \ldots, D_r \\ T_1, T_2, \ldots, T_r}} \E \norm{ \E_{G_0} f(G_0) - \E_U f(U) }.
	\end{align*}

	Now we must show that the function $f$ is fooled by $G_0$ for most values of $D_i$ and $T_i$. 
	Luckily, for $r$ large enough, $f$ is often a constant function and therefore fooled by any distribution. 

	In particular, let $T_i[j]$ denote the $j$-th bit of $T_i$ and define the indicator random variables 
	$$Y_j = \bigwedge\limits_{i=1}^r T_i[j].$$
	Note that $f(x)$ depends on the $j$-th bit of $x$ only if $Y_j = 1$, and so $f(x)$ is constant if $\sum_{j=1}^n Y_j = 0$.
	Also note that $\prob{Y_j = 1} = 2^{-r}$.
	By applying a Markov inequality, we have

	\begin{align*}
		\E_{\substack{D_1, D_2, \ldots, D_r \\ T_1, T_2, \ldots, T_r}} \norm{ \E_{G_0} f(G_0) - \E_U f(U) }
		&\leq \prob{\sum_{j=1}^n Y_j \geq 1} \underset{\substack{D_1, D_2, \ldots, D_r \\ T_1, T_2, \ldots, T_r}} \max \norm{ \E_{G_0} f(G_0) - \E_U f(U) } \\
		&\leq \E\left[\sum_{j=1}^n Y_j \right] \cdot 2 w^{1/2} \\
		&= \frac{2 n w^{1/2}}{2^r}
		\;.
	\end{align*}

	If we set $k \geq 5 \lg n + 2 \lg w$ and $r \geq 2 \lg n + \frac{1}{2} \lg w$, $G_r$ fools $F$ with error
	\[
		\varepsilon = \norm{\E_{G_r} F(G_r) - \E_U F(U)} \leq r \cdot \frac{nw}{2^{k/2}} + \frac{2 n w^{1/2}}{2^r} \leq \frac{3}{n}
		\;.
		\qedhere
	\]
\end{proof}

\section{Proof of \autoref{mainlem2}}

The proof of \autoref{mainlem2} follows the same story as the previous section with details differing in two places.
First, we derive an analogue of \autoref{prooflem1} which is slightly messier due to our now weaker pseudorandom primitives. 
Secondly, in order to get the best possible seed-length in the small error regime, this time we analyze the base case with a bit more care. In particular, we upgrade the Markov argument to a Chernoff bound for $\gamma$-almost $k$-wise independent variables. 

\begin{lemma} 
	Let
	\[
		H(x) = \sum_{\substack{\alpha \in \F_2^n \\ |\alpha| = k}} \hat{H}_{\alpha} \ch(x)
	\]
	be some function whose Fourier expansion is supported only at degree $k$. 
	Let $D$, $T$, and $U$  denote respectively a $\delta$-biased, a $\gamma$-almost $k$-wise independent, and a uniform distribution over $\F_2^n$.
	Then we have
	\[
		\E_{D,T} \| \E_U H(D + T \land U) \|^2
		\leq \Big( 2^{-k} + \gamma \Big) \left( \delta \cdot \Big( \sum_{|\alpha| = k} \|\hat{H}_\alpha \| \Big)^2 
			+\sum_{|\alpha| = k} \|\hat{H}_\alpha \|^2 \right)
		\;.
	\]
\end{lemma}
\begin{proof}
	As before, we have
	\[
		\| \E_U H(D + T \land U) \|^2 
		= \sum_{|\alpha| = k} \sum_{|\beta| = k} \ip{ \hat{H}_\alpha}{ \hat{H}_\beta} \chi_{\alpha + \beta}(D) \ind{ (\alpha \lor \beta)  \land T= 0}
		\;.
	\]
	We analyze the terms with $\alpha = \beta$ and $\alpha \neq \beta$ separately. 
	For the cross terms we have
	\begin{align*}
		\E_D
		&\E_T \sum_{|\alpha| = k} \sum_{\beta \neq \alpha}\ip{ \hat{H}_\alpha}{ \hat{H}_\beta} \chi_{\alpha + \beta}(D) \ind{(\alpha \lor \beta)  \land T = 0} \\
		&=  \sum_{|\alpha| = k}\sum_{\beta \neq \alpha} \ip{ \hat{H}_\alpha}{ \hat{H}_\beta} \Big( \E_D \chi_{\alpha + \beta}(D) \Big) \Big(  \E_T \ind{(\alpha \lor \beta)  \land T  = 0} \Big) \\
		&\leq  \sum_{|\alpha| = k}\sum_{\beta \neq \alpha} \norm{ \hat{H}_\alpha} \norm{ \hat{H}_\beta} \cdot \delta \cdot   \Big( 2^{-k} + \gamma \Big)  \\
		&\leq  \delta \cdot  \Big( 2^{-k} + \gamma \Big) \Big(  \sum_{|\alpha| = k} \norm{ \hat{H}_\alpha}  \Big)^2
		\;.
	\end{align*}

	For the like terms we have
	\begin{align*}
	\E_T
		&\E_D \sum_{|\alpha| = k} \ip{ \hat{H}_\alpha}{ \hat{H}_\alpha} \chi_0(D) \ind{T \land \alpha = 0} \\
		&= \sum_{|\alpha| = k} \norm{\hat{H}_\alpha}^2 \cdot \E_T \ind{T \land \alpha = 0} \\
		&\leq\sum_{|\alpha| = k} \norm{\hat{H}_\alpha}^2 \cdot \Big(2^{-k} + \gamma \Big)
		\;.
		\qedhere
	\end{align*}
\end{proof}

We now derive the following analogue of \autoref{prooflem2}.

\begin{lemma} 
Suppose $F : \F_2^n \rightarrow \R^{w \times w}$ encodes a branching program.
Let $D$, $T$, and $U$  denote respectively a $\delta$-biased, a $\gamma$-almost $k$-wise independent, and a uniform distribution over $\F_2^n$,
and let
Then $D + T \land U$ fools $F$ with error
$$ \varepsilon =  \norm{\E_{D,T,U}F(D + T \land U) - \E_U F(U)} \leq \left( \sqrt{\delta} \mathcal{L}(n,w;k) + \left(\tfrac{1}{2}\right)^{k/2} + \sqrt{\gamma}  \right) \cdot nw.$$
\end{lemma}

\begin{proof}
	Again we use \autoref{prop1} to split $F$ into high and low degree components:

	\begin{align*}
	\| \E_{D,T,U} F(D + T \land U) - \E_{U} F(U) \| &\\
	\leq  \| \E_{D,T,U} L(D + T \land U)  \|& + \sum_{i=1}^n \| \E_{D,T,U} H_i(D + T \land U) F^{> i}(D + T \land U) \|.
	\end{align*}

	For the low-degree component we have

	\begin{align*}
	\norm{ \E_{D,T,U} F(D + T \land U) }
	&= \norm{ \E_D \E_T \E_U \sum_{\substack{\alpha \in \F_2^n \\ 0 < |\alpha| < k}} \hat{F}_\alpha \chi_\alpha(D + T \land U) } \\
	&= \norm{ \E_T \E_U \sum_{\substack{\alpha \in \F_2^n \\ 0 < |\alpha| < k}} \hat{F}_\alpha \Big( \E_D \ch(D) \Big) \ch(T \land U)  }\\
	&\leq \sum_{\substack{\alpha \in \F_2^n \\ 0 < |\alpha| < k}} \norm{ \hat{F}_\alpha } \cdot \delta \\
	&= \delta \sum_{i=1}^{k-1} \mathcal{L}_i(F).
	\end{align*}

	Now we proceed as before.

	\begin{align*}
		\| \E_{D,T,U} F(D + T \land U) - &\E_{U} F(U) \| \\
		&\le\delta \sum_{i=1}^{k-1} \mathcal{L}_i(F) + \sum_{i=1}^n \| \E_{D,T,U} H_i(D + T \land U) F^{> i}(D + T \land U) \| \\
		&\le\delta \sum_{i=1}^{k-1} \mathcal{L}_i(F) + \sum_{i=1}^n  \Big( \E_{D,T} \|\E_U H_i(D + T \land U) \|^2 \Big)^{1/2} w^{1/2} \\
		&\le\delta \sum_{i=1}^{k-1} \mathcal{L}_i(F)  + \sum_{i=1}^n \Big( \delta \mathcal{L}_k(F^{\leq i})^2 + \left(\tfrac{1}{2}\right)^k + \gamma \Big)^{1/2} w \\
		&\le\left( \sqrt{\delta} \mathcal{L}(n,w;k) + \left(\tfrac{1}{2}\right)^{k/2} + \sqrt{\gamma}  \right) \cdot nw
		\;.
		\qedhere
	\end{align*}
\end{proof}

\subsection{Proof of \autoref{mainlem2}}

\begin{proof}
	We proceed as in the proof of \autoref{mainlem1}, except this time we derive a sharper bound on the quantity
	$$ \underset{\substack{D_1, D_2, \ldots, D_r \\ T_1, T_2, \ldots, T_r}} \E \norm{ \E_{G_0^*} f(G_0^*) - \E_U f(U) }. $$
	Again, define the random variable
	$$ Y = T_1 \land T_2 \land \cdots \land T_r.$$
	Recall that $G_0^*$ is a $320 k$-wise independent distribution, so if $|Y| \leq 320 k$ then 
	$$ \norm{ \E_{G_0^*} f(G_0^*) - \E_U f(U) } = 0.$$
	Therefore

	\begin{align*}
	\underset{\substack{D_1, D_2, \ldots, D_r \\ T_1, T_2, \ldots, T_r}} \E \norm{ \E_{G_0^*} f(G_0^*) - \E_U f(U) }
	&\leq \prob{|Y| \geq 320 k} \underset{\substack{D_1, D_2, \ldots, D_r \\ T_1, T_2, \ldots, T_r}} \max \norm{ \E_{G_0^*} f(G_0^*) - \E_U f(U) } \\
	&\leq \prob{|Y| \geq 320 k} \cdot 2 w^{1/2}.
	\end{align*}

	We appeal to the following extension of the Chernoff bound for $k$-wise independent variables.
	\begin{lemma}{\textnormal{(see \cite{svw14}, Lemma A.1)} }
	Suppose $X_1, X_2, \ldots X_t$ are $\gamma$-almost $k$-wise independent variables with $X_i \in \{0,1\}$. Then
	$$ \prob{\frac{1}{t}\sum_{i=1}^t X_i \geq \frac{1}{2} + \frac{a}{2}} \leq \left(\frac{40k}{a^2 t} \right)^{\lfloor k/2 \rfloor} + 2 \gamma \left( \frac{2}{a}\right)^k.$$
	\end{lemma}

	As a result, if $T$ is a $\gamma$-almost $k$-wise independent distribution and $\alpha$ is any fixed bitmask with hamming weight $|\alpha| \geq 320k$, we have
	$$ \prob{|\alpha \land T| \geq \tfrac{3}{4} |\alpha|} \leq \left( \tfrac{1}{2} \right)^{\lfloor k/2 \rfloor} + 2 \gamma \cdot 4^k.$$
	Noting that $ \left(\frac{3}{4}\right)^r n \leq 1$, a simple union bound argument shows that
	$$ \prob{|T_1 \land T_2 \land \cdots \land T_r| \geq 320 k} \leq r \cdot \left( \left( \tfrac{1}{2} \right)^{\lfloor k/2 \rfloor} + 2 \gamma \cdot 4^k \right) .$$

	To conclude, we have
	\begin{align*}
		\norm{\E_{G_r^*} F(G_r^*) - \E_U F(U)} 
		&\leq r \cdot \left( \sqrt{\delta} \mathcal{L}(n,w;k) + \left(\tfrac{1}{2}\right)^{k/2} + \sqrt{\gamma}  \right) \cdot nw 
		+ r \cdot \left( \left( \tfrac{1}{2} \right)^{\lfloor k/2 \rfloor} + 2 \gamma \cdot 4^k \right) \cdot 2w^{1/2} \\
		&\leq O\paren{ \left( \sqrt{\delta} \mathcal{L}(n,w;k) + \left(\tfrac{1}{2}\right)^{k/2} + \sqrt{\gamma} + \gamma 4^k \right) \cdot nwr}
		\;.
		\qedhere
	\end{align*}
\end{proof}

\bibliography{prg-hlv}{}
\bibliographystyle{alpha}

\end{document}